%% file: main.tex
\documentclass[a4paper]{article}
\pdfoutput=1
\usepackage[margin=1in]{geometry}

\usepackage{microtype}
\usepackage[algoruled,linesnumbered,algo2e,vlined]{algorithm2e}
\usepackage{colortbl}
\usepackage{amsthm}
\usepackage{amsmath}

\newtheorem{theorem}{Theorem}
\newtheorem{definition}[theorem]{Definition}
\newtheorem{lemma}[theorem]{Lemma}

\newcommand{\tgray}{\cellcolor[gray]{0.85}}
\newcommand{\Lpref}{L^{pref}}
\newcommand{\Lsuf}{L^{suf}}
\newcommand{\EDpref}{\mathit{ED}^{pref}}
\newcommand{\Lmax}{L_{\max}}
\newcommand{\Lsum}{L_{\mathrm{sum}}}
\newcommand{\Gmin}{G_{\min}}
\newcommand{\RLE}{\mathit{RLE}}

\title{Faster STR-IC-LCS computation via RLE}

\date{}

\author{Keita Kuboi\quad
        Yuta Fujishige\quad
        Shunsuke Inenaga\quad
        Hideo Bannai\quad
        Masayuki Takeda\\
        {Department of Informatics, Kyushu University, Japan}\\
        {\texttt{\{keita.kuboi,yuta.fujishige,inenaga,bannai,takeda\}@inf.kyushu-u.ac.jp}}}

\begin{document}

\maketitle

\input{0_abstract}
\input{1_introduction}
\input{2_preliminaries}
\input{3_algorithm}
\input{4_conclusion}

\clearpage
\bibliography{ref}
\bibliographystyle{plain}

\clearpage
\appendix
\input{5_appendix}

\end{document}

%% file: 0_abstract.tex
\begin{abstract}
The constrained LCS problem asks one to find a longest common subsequence
of two input strings $A$ and $B$ with some constraints.
The STR-IC-LCS problem is a variant of the constrained LCS problem,
where the solution must include a given constraint string $C$ as a substring.
Given two strings $A$ and $B$ of respective lengths $M$ and $N$,
and a constraint string $C$ of length at most $\min\{M, N\}$,
the best known algorithm for the STR-IC-LCS problem,
proposed by Deorowicz~({\em Inf. Process. Lett.}, 11:423--426, 2012),
runs in $O(MN)$ time.
In this work, we present an $O(mN + nM)$-time solution to
the STR-IC-LCS problem, where $m$ and $n$ denote the sizes of
the run-length encodings of $A$ and $B$, respectively.
Since $m \leq M$ and $n \leq N$ always hold, our algorithm is always
as fast as Deorowicz's algorithm, and is faster when input strings are compressible via RLE.
\end{abstract}

%% file: 1_introduction.tex
\section{Introduction}
{\em Longest common subsequence} (LCS) is one of the most basic
measures of similarity between strings, and there is a vast amount of literature concerning its efficient computation.
An LCS of two strings $A$ and $B$ of lengths $M$ and $N$, respectively, is a longest string that is a subsequence of both $A$ and $B$.
There is a well known $O(MN)$ time and space dynamic programming (DP) algorithm~\cite{DBLP:journals/jacm/WagnerF74} to compute
an LCS between two strings.
LCS has applications in bioinformatics~\cite{korkin02:_multip,wang11}, file comparisons~\cite{hunt76:_algor_differ_file_compar,heckel78}, pattern recognition~\cite{stern13:_most}, etc.

Recently, several variants of the problem
which try to find a longest common subsequence that satisfy some constraints have
been considered. 
In 2003, Tsai~\cite{DBLP:journals/ipl/Tsai03} proposed the constrained LCS (CLCS) problem,
where, given strings $A,B$ with respective lengths $M,N$, and a constraint string $C$ of length $K$, the problem is to find a longest string that contains $C$ as a subsequence
and is also a common subsequence of $A$ and $B$.
Tsai gave an $O(M^2N^2K)$ time solution, which was improved in 2004 by Chin et al. to $O(MNK)$ time~\cite{DBLP:journals/ipl/ChinSFHK04}.
Variants of the constrained LCS problem called
SEQ-IC-LCS, SEQ-EC-LCS, STR-IC-LCS, and STR-EC-LCS,
were considered by Chen and Chao in 2011~\cite{DBLP:journals/jco/ChenC11}.
Each problem considers as input, three strings $A,B$ and $C$,
and the problem is to find a longest string that includes (IC) or excludes (EC) 
$C$ as a subsequence (SEQ) or substring (STR) and is a common subsequence of $A$ and $B$
(i.e., CLCS is equivalent to the SEQ-IC-LCS problem).
The best solution for each of the problems is shown in Table~\ref{tbl:timecomplexityGCLCS}.

\begin{table}[htbp]
  \begin{center}
    \caption{Time complexities of best known solutions to various constrained LCS problems.}
    \label{tbl:timecomplexityGCLCS}
    \begin{tabular}{ccc} \hline
      Problem & DP solution & DP solution using RLE \\ \hline
      SEQ-IC-LCS & $O(MNK)$~\cite{DBLP:journals/ipl/ChinSFHK04}
      & $O(M+N+K\min\{mN, nM\})$~\cite{DBLP:journals/cj/LiuWC15} \\
      SEQ-EC-LCS & $O(MNK)$~\cite{DBLP:journals/jco/ChenC11} & - \\
      STR-IC-LCS & $O(MN)$~\cite{DBLP:journals/ipl/Deorowicz12} & $O(mN+nM)$ [this work] \\
      STR-EC-LCS & $O(MNK)$~\cite{DBLP:journals/ipl/WangWWZ13} & - \\ \hline
    \end{tabular}
  \end{center}
\end{table}

In order to speed up the LCS computation, one direction of research that has received much attention is
to apply compression, namely, run-length encoding (RLE) of strings.
Bunke and Csirik~\cite{DBLP:journals/ipl/BunkeC95} were one of the first to consider such a scenario, and
proposed an $O(mN+nM)$ time algorithm. Here, $m,n$ are the sizes of the RLE of the input strings of lengths $M$ and $N$, respectively.
Notice that since RLE can be computed in linear time, and $m \leq M$ and $n \leq N$, the algorithm is always asymptotically faster than
the standard $O(NM)$ time dynamic programming algorithm, especially when the strings are compressible by RLE.
Furthermore, Ahsan et al. proposed an algorithm which runs in $O((m+n)+R\log\log(mn)+R\log\log(M+N))$ time~\cite{DBLP:journals/jcp/AhsanAR14},
where $R$ is the total number of pairs of runs of the same character in the two RLE strings, i.e. $R \in O(mn)$,
and the algorithm can be much faster when the strings are compressible by RLE.

For the constrained LCS problems, RLE based solutions for only the SEQ-IC-LCS problem have been proposed.
In 2012, an $O(K(mN+nM))$ time algorithm was proposed by Ann et al.~\cite{DBLP:journals/tcs/AnnYTH12}.
Later, in 2015, Liu et al. proposed a faster $O(M+N+K\min\{mN,nM\})$ time algorithm~\cite{DBLP:journals/cj/LiuWC15}.

In this paper, we present the first RLE based solution for the STR-IC-LCS problem that runs in $O(mN+nM)$ time.
Again, since RLE can be computed in linear time, and $m \leq M$ and $n \leq N$, the proposed algorithm is always asymptotically faster than
the best known solution for the STR-IC-LCS problem by Deorowicz~\cite{DBLP:journals/ipl/Deorowicz12}, which runs in $O(MN)$ time.

A common criticism against RLE based solutions is a claim that, although they are theoretically interesting, 
since most strings ``in the real world'' are not compressible by RLE, their applicability is limited and they are only useful in extreme artificial cases.
We believe that this is not entirely true.
There can be cases where RLE is a natural encoding of the data, for example, in music, a melody
can be expressed as a string of pitches and their duration.
Furthermore, in the data mining community, there exist popular preprocessing schemes for analyzing various types of
time series data, which convert the time series to strings over a fairly small alphabet as an approximation of the original data,
after which various analyses are conducted (e.g. SAX (Symbolic Aggregate approXimation)~\cite{lin07:_exper_sax},
{\em clipped} bit representation~\cite{bagnall06:_bit_level_repres_time_series}, etc.).
These conversions are likely to produce strings which are compressible by RLE
(and in fact, shown to be effective in~\cite{bagnall06:_bit_level_repres_time_series}),
indicating that RLE based solutions may have a wider range of application than commonly perceived.

%% file: 2_preliminaries.tex
\section{Preliminaries}
Let $\Sigma$ be the finite set of characters, and $\Sigma ^\ast$ be the set of strings.
For any string $A$, let $|A|$ be the length of $A$.
For any $1 \leq i \leq i' \leq |A|$, let $A[i]$ be the $i$th character of $A$ and let $A[i..i'] = A[i] \cdots A[i']$ denote a substring of $A$.
Especially, $A[1..i']$ denotes a {\em prefix} of $A$, and $A[i..|A|]$ denotes a {\em suffix} of $A$.
A string $Z$ is a {\em subsequence} of $A$ if $Z$ can be obtained from $A$ by removing zero or more characters.
For two string $A$ and $B$, a string $Z$ is a {\em longest common subsequence} (LCS) of $A,B$, 
if $Z$ is a longest string that is a subsequence of both $A$ and $B$.
For any $1 \leq i \leq |A|$ and $1 \leq j \leq |B|$, let $\Lpref(i,j)$ denote the length of an LCS of $A[1..i]$, $B[1..j]$, and let $\Lsuf(i,j)$ denote the length of an LCS of $A[i..|A|]$, $B[j..|B|]$.
The LCS problem is to compute the length of an LCS of given two strings $A$ and $B$.
A well known solution is dynamic programming,
which computes in $O(MN)$ time, a table (which we will call {\em DP table}) of size $O(MN)$ that stores values of $\Lpref(i,j)$ for all $1 \leq i \leq M$, $1 \leq j\leq N$.
The DP table for $\Lsuf(i,j)$ can be computed similarly.

For two strings $A,B$ and a constraint string $C$, a string $Z$ is an STR-IC-LCS of $A,B,C$, 
if $Z$ is a longest string that includes $C$ as a substring
and also is a subsequence of both $A$ and $B$.
The STR-IC-LCS problem is to compute the length of an STR-IC-LCS of any given three strings $A$, $B$ and $C$.
For example, if $A=\mathtt{abacab}$, $B=\mathtt{babcaba}$, $C=\mathtt{bb}$, then $\mathtt{abcab}$ and $\mathtt{bacab}$ are LCSs of $A,B$,
and $\mathtt{abb}$ is an STR-IC-LCS of $A,B,C$. 

The run-length encoding (RLE) of a string $A$ is a kind of compressed representation of $A$ where each maximal run of the same character is represented by a pair of the character and the length of the run.
Let $\RLE(A)$ denote the RLE of a string $A$. The {\em size} of $\RLE(A)$ is the number of the runs in $A$, and is denoted by $|\RLE(A)|$.
By definition, $|\RLE(A)|$ is always less than or equal to $|A|$.

In the next section, we consider the STR-IC-LCS problem of strings $A$, $B$ and constraint string $C$.
Let $|A|=M$, $|B|=N$, $|C|=K$, $|\RLE(A)|=m$ and $|\RLE(B)|=n$.
We assume that $K \leq \min(M, N)$ and $|\RLE(C)| \leq \min(m,n)$, since in such case there can be no solution.
We also assume that $K > 0$, because in that case the problem becomes 
the normal LCS problem of $A,B$.

%% file: 3_algorithm.tex
\section{Algorithm}
In this section, we will first introduce a slightly modified version of
Deorowicz's $O(MN)$-time algorithm for the STR-IC-LCS
problem~\cite{DBLP:journals/ipl/Deorowicz12}, and then
propose our $O(mN + nM)$-time algorithm which is
based on his dynamic programming approach but uses RLE.

\subsection{Deorowicz's $O(MN)$ Algorithm}\label{ssec:DeorowiczAlgo}
We first define the notion of minimal $C$-intervals of a string.
\begin{definition}\label{def:mincinterval}
  For any strings $A$ and $C$, an interval $[s, f]$ is a {\em minimal $C$-interval of $A$} if
  \begin{itemize}
    \item $C$ is a subsequence of $A[s..f]$, and
    \item $C$ is not a subsequence of $A[s+1..f]$ or $A[s..f-1]$.
  \end{itemize}
\end{definition}

Deorowicz's algorithm is based on Lemma~\ref{lem:CharacteristicSTR-IC-LCS}, which is used implicitly in~\cite{DBLP:journals/ipl/Deorowicz12}.
\begin{lemma}[implicit in~\cite{DBLP:journals/ipl/Deorowicz12}]\label{lem:CharacteristicSTR-IC-LCS}
  If $Z$ is an STR-IC-LCS of $A, B, C$, then
  there exist minimal $C$-intervals $[s, f], [s', f']$
  $(1 \leq s \leq f \leq M$, $1 \leq s' \leq f' \leq N)$ respectively of $A$ and $B$,
  such that $Z = XCY$, where
  $X$ is an LCS of $A[1..s-1]$ and $B[1..s'-1]$ and
  $Y$ is an LCS of $A[f+1..M]$ and $B[f'+1..N]$.
\end{lemma}

\begin{proof}
  From the definition of STR-IC-LCS, $C$ is a substring of $Z$, and therefore,
  there exist (possibly empty) strings $X,Y$ such that $Z = XCY$.
  Also, since $Z$ is a common subsequence of $A$ and $B$, there exist
  monotonically increasing sequences $i_1, \ldots, i_{|Z|}$ and $j_1, \ldots, j_{|Z|}$
  such that
  $Z = A[i_1] \cdots A[i_{|Z|}]$ $= B[j_1] \cdots B[j_{|Z|}]$, and
  $C = A[i_{|X|+1}] \cdots A[i_{|X|+K}]$ $= B[j_{|X|+1}] \cdots B[j_{|X|+K}]$.

  Now, since $C$ is a subsequence of
  $A[i_{|X|+1}..i_{|X|+K}]$ and $B[j_{|X|+1}..j_{|X|+K}]$ there exist
  minimal $C$-intervals $[s, f]$, $[s', f']$ respectively of $A$ and $B$ that satisfy
  $i_{|X|+1} \leq s \leq f \leq i_{|X|+K}$ and
  $j_{|X|+1} \leq s' \leq f' \leq j_{|X|+K}$.
  Let $X'$ be an LCS of $A[1..s-1]$ and $B[1..s'-1]$, 
  and $Y'$ an LCS of $A[f+1..M]$ and $B[f'+1..N]$.
  Since $X$ must be a common subsequence of $A[1..s-1]$ and $B[1..s'-1]$, 
  and $Y$ a common subsequence of $A[f+1..M]$ and $B[f'+1..N]$,
  we have $|X'| \geq |X|$ and $|Y'| \geq |Y|$.  
  However, we cannot have that $|X'| > |X|$ or $|Y'| > |Y|$
  since otherwise, $X'CY'$ would be a string longer than $Z$ that contains $C$ as a substring,
  and is a common subsequence of $A,B$, contradicting that
  $Z$ is an STR-IC-LCS of $A,B,C$.
  Thus, $|X|=|X'|$ and $|Y|=|Y'|$ implying that $X$ is also an LCS of $A[1..s-1], B[1..s'-1]$,
  and $Y$ is also an LCS of $A[f+1..M], B[f'+1..N]$, proving the lemma.
\end{proof}
The algorithm consists of the following two steps, whose correctness follows from Lemma \ref{lem:CharacteristicSTR-IC-LCS}.
\begin{description}
  \item[Step 1] Compute all minimal $C$-intervals of $A$ and $B$.
  \item[Step 2] For all pairs of a minimal $C$-interval $[s,f]$ of $A$
    and a minimal $C$-interval $[s',f']$ of $B$,
    compute
    the length of an LCS of the corresponding prefixes of $A$ and $B$ (i.e., $\Lpref(s-1,s'-1)$)
    and that of the corresponding suffixes of $A$ and $B$ (i.e., $\Lsuf(f+1,f'+1)$).
    The largest sum of LCS lengths plus $|C|$ (i.e., $\Lpref(s-1,s'-1) + \Lsuf(f+1,f'+1) +|C|$) is the length of an STR-IC-LCS.
\end{description}

The steps can be executed in the following running times.
For Step 1,
there are respectively at most $M$ and $N$ minimal $C$-intervals of $A$ and $B$,
which can be enumerated in $O(MK)$ and $O(NK)$ time.
For Step 2, we precompute, in $O(MN)$ time,
two dynamic programming tables which respectively contain the values of $\Lpref(i,j)$ and $\Lsuf(i,j)$ for each $1\leq i \leq M$ and $1 \leq j \leq N$.
Using these tables, the value $\Lpref(s-1,s'-1) + \Lsuf(f+1,f'+1) +|C|$ can be computed in constant time for any $[s,f]$ and $[s',f']$.
There are $O(MN)$ possible pairs of minimal $C$-intervals, so Step 2 can be done in $O(MN)$ time.
In total, since $K \leq M, K \leq N$, the STR-IC-LCS problem can be solved in $O(MN)$ time.

We note that in the original presentation of Deorowicz's algorithm, 
{\em right}-minimal $C$-intervals, that is, intervals $[s,f]$ where $C$ is a subsequence of $A[s..f]$ but not of $A[s..f-1]$
are computed,
instead of minimal $C$-intervals as defined in Definition~\ref{def:mincinterval}.
Although the number of considered intervals changes, this does not influence the asymptotic complexities in the non-RLE case.
However, as we will see in Lemma~\ref{lem:MinInterval k>1} of Section~\ref{ssec:OurAlgo}, this is an essential difference for the RLE case, since, when $|\RLE(C)|>1$, 
the number of minimal $C$-intervals of $A$ and $B$ can be bounded by $O(m)$ and $O(n)$, but
the number of right-minimal $C$-intervals of $A$ and $B$ cannot, and are only bounded by $O(M)$ and $O(N)$.

\subsection{Our Algorithm via RLE}\label{ssec:OurAlgo}
In this subsection, we propose an efficient algorithm based on Deorowicz's algorithm explained in
Subsection~\ref{ssec:DeorowiczAlgo}, extended to strings expressed in RLE.
There are two main cases to consider:
when $|\RLE(C)| = 1$, i.e., when $C$ consists of only one type of character,
and when $|\RLE(C)| > 1$, i.e., when $C$ contains at least two different characters.

\subsubsection{Case $|\RLE(C)| > 1$}

\begin{theorem}\label{thm:TimeComplexity k>1}
  Let $A,B,C$ be any strings and let $|A|=M$, $|B|=N$, $|\RLE(A)|=m$ and $|\RLE(B)|=n$.
  If $|\RLE(C)|>1$, we can compute the length of an STR-IC-LCS of $A,B,C$ in $O(mN + nM)$ time.
\end{theorem}

For Step 1, we execute the following procedure to enumerate all minimal $C$-intervals of $A$ and $B$.
Let $s_0 = 0$.
First, find the right minimal $C$-interval starting at $s_0+1$, i.e., the smallest position $f_1$ such that $C$ is a subsequence of $A[s_0+1..f_1]$.
Next, starting from position $f_1$ of $A$, search backwards to find the left minimal $C$-interval ending at $f_1$, i.e., 
the largest position $s_1$ such that $C$ is a subsequence of $A[s_1..f_1]$.
The process is then repeated, i.e., find the smallest position $f_2$ such that $C$ is a subsequence of $A[s_1+1..f_2]$,
and then search backwards to find the largest position $s_2$ such that $C$ is a subsequence of $A[s_2..f_2]$, and so on.
It is easy to see that the intervals
$[s_1, f_1], [s_2, f_2], \ldots$
obtained by repeating this procedure until reaching the end of $A$ are all the
minimal $C$-intervals of $A$,
since each interval that is found is distinct, and there cannot exist another minimal $C$-interval between those found by the procedure.
The same is done for $B$.
For non-RLE strings, this takes $O((M+N)K)$ time.
The Lemma below shows that the procedure can be implemented more efficiently using RLE.

\begin{lemma}\label{lem:MinInterval k>1}
  Let $A$ and $C$ be strings where $|A| = M$, $|\RLE(A)| = m$ and $|C| = K$.
  If $|\RLE(C)| > 1$,
  the number of minimal $C$-intervals of $A$ is $O(m)$ and can be enumerated in $O(M+mK)$ time.
\end{lemma}

\begin{proof}
  Because $|\RLE(C)|>1$, it is easy to see from the backward search in the procedure described above,
  that for any minimal $C$-interval of $A$, 
  there is a unique run of $A$ such that the last character of the first run of $C$ corresponds to the last character of that run.
  Therefore, the number of minimal $C$-intervals of $A$ is $O(m)$.

  We can compute $\RLE(A) = a_1^{M_1} \cdots a_m^{M_m}$ and $\RLE(C)=c_1^{K_1} \cdots c_k^{K_k}$ in $O(M+K)$ time.
  What remains is to show that the forward/backward search procedure described 
  above to compute all minimal $C$-intervals of $A$ can be implemented in $O(mK)$ time.
  The pseudo-code of the algorithm described is shown in Algorithm~\ref{alg:minCinterval}.
  
  In the forward search, we scan $\RLE(A)$ to find a right minimal $C$-interval by greedily matching the runs of $\RLE(C)$ to $\RLE(A)$.
  We maintain the character $c_q$ and exponent $\mathit{rest}$ of the first run $c_q^{\mathit{rest}}$ of $\RLE(C')$, where $C'$ is the suffix of $C$ that is not yet matched.
  When comparing a run $a_p^{M_p}$ of $\RLE(A)$ and $c_q^{\mathit{rest}}$,
  if the characters are different (i.e., $a_p \neq c_q$), we know that the entire run $a_p^{M_p}$ will not match
  and thus we can consider the next run of $A$.
  Suppose the characters are the same. Then, if $M_p < \mathit{rest}$, the entire run $a_p^{M_p}$ of $A$ is matched, and we can consider the next run $a_{p+1}^{M_{p+1}}$ of $A$.
  Also, $\mathit{rest}$ can be updated accordingly in constant time by simple arithmetic. Furthermore, since $c_q = a_p \neq a_{p+1}$, we can in fact skip to the next run
  $a_{p+2}^{M_{p+2}}$.
  If $M_p \geq \mathit{rest}$, the entire run $c_q^{\mathit{rest}}$ is matched, and we consider the next run $c_{q+1}^{K_{q+1}}$ in $C$.
  Also, since $a_p = c_q \neq c_{q+1}$, we can skip the rest of $a_p^{M_p}$ and consider the next run $a_{p+1}^{M_{p+1}}$ of $A$.
  Thus, we spend only constant time for each run of $A$ that is scanned in the forward search. The same holds for the backward search.
  
  To finish the proof, we show that the total number of times that each run of $A$ is scanned in the procedure is bounded by $O(K)$, i.e.,
  the number of minimal $C$-intervals of $A$ that intersects with a given run $a_p^{M_p}$ of $A$ is $O(K)$.
  Since $|\RLE(C)| > 1$, a minimal $C$-interval cannot be contained in $a_p^{M_p}$.
  Thus, for a minimal $C$-interval to intersect with the run $a_p^{M_p}$, it must cross either the left boundary of the run, or the right boundary of the run.
  For a minimal $C$-interval to cross the left boundary of the run, it must be that for some non-empty strings $u,v$ such that $C = uv$,
  $u$ occurs as a subsequence in $a_1^{M_1} \cdots a_{p-1}^{M_{p-1}}$ and 
  $v$ occurs as a subsequence in $a_p^{M_p} \cdots a_{m}^{M_{m}}$.
  The minimal $C$-interval corresponds to the union of the left minimal $u$-interval ending at the
  left boundary of the run and the right minimal $v$-interval starting at the left boundary of the run and is thus unique for $u,v$.
  Similar arguments also hold for minimal $C$-intervals that cross the right boundary of $a_p^{M_p}$.
   Since there are only $K-1$ choices for $u,v$, the claim holds, thus proving the Lemma.
\end{proof}

\begin{algorithm2e}
  \caption{computing all minimal $C$-intervals of $A$}
  \label{alg:minCinterval}
  \KwIn{strings $A$ and $C$}
  \KwOut{all minimal $C$-intervals $[s_1, f_1], \ldots, [s_l, f_l]$ of $A$}
  \tcp{$\RLE(A) = a_1^{M_1} \cdots a_m^{M_m}$,
    $\RLE(C) = c_1^{K_1} \cdots c_k^{K_k}$}
  \tcp{$M_{1..p} = M_1 + \cdots + M_p$}
  \tcp{$p,q$ : index of run in $A,C$ respectively}
  \tcp{$rest$ : number of rest of searching characters of $c_q^{K_q}$}
  \tcp{$l$ : number of minimal $C$-intervals in $A$}
  
  $p \gets 1$;
  $q \gets 1$;
  $rest \gets K_1$;
  $l \gets 0$\;
  
  \While{true}{
    \While(// forward search){$p \leq m$ and $q \leq k$}{ \label{algln:fsearch_begin}
      \lIf{$a_p \neq c_q$}{$p \gets p+1$\;} \label{algln:ap!=cq}
      \Else{
        \uIf{$M_p \geq rest$}{ \label{algln:ap=cq&Mp>=rest_begin}
          $q \gets q+1$\;
          \lIf{$q > k$}{$l \gets l+1$; $f_l \gets M_{1..p-1} + rest$\;}
          \lElse{$p \gets p+1$; $rest \gets K_q$\;} \label{algln:ap=cq&Mp>=rest_end}
        }
        \lElse{$rest \gets rest - M_p$; $p \gets p+2$\;} \label{algln:fsearch_end}
      }
    }
    \lIf{$p > m$}{{\bf break}\;} \label{algln:algoend}
    
    $p \gets p-1$\;
    \lIf{$rest = K_k$}{$q \gets q-1$; $rest \gets K_{k-1}$\;}
    \lElse{$q \gets k$; $rest \gets K_k - rest$\;}
    
    \While(// backward search){$q \geq 1$}{ \label{algln:bsearch_begin}
      \lIf{$a_p \neq c_q$}{$p \gets p-1$\;}
      \Else{
        \uIf{$M_p \geq rest$}{
          $q \gets q+1$\;
          \lIf{$q < 1$}{$s_l \gets M_{1..p} - rest + 1$\;}
          \lElse{$p \gets p-1$; $rest \gets K_q$\;}
        }
        \lElse{$rest \gets rest - M_p$; $p \gets p-2$\;} \label{algln:bsearch_end}
      }
    }
    
    $p \gets p+1$;
    $q \gets 1$;
    $rest \gets K_1 - rest + 1$\;
  }
  \Return{$[s_1, f_1], \ldots, [s_l, f_l]$}\;
\end{algorithm2e}

In Deorowicz's algorithm, two DP tables were computed for Step 2, which took $O(MN)$ time.
For our algorithm, we use a compressed representation of the DP table for $A$ and $B$,
proposed by Bunke and Csirik~\cite{DBLP:journals/ipl/BunkeC95}, instead of the normal DP table.
We note that Bunke and Csirik actually solved the edit distance problem 
when the cost is $1$ for insertion and deletion, and $2$ for substitution,
but this easily translates to LCS: $\Lpref(i,j) = (i+j-\EDpref(i,j))/2$, where
$\EDpref(i,j)$ denotes the edit distance with such costs, between $A[1..i]$ and $B[1..j]$.

\begin{definition}[\cite{DBLP:journals/ipl/BunkeC95}]\label{def:compressedDPtable}
  Let $A$, $B$ be strings of length $M$, $N$ respectively, where $\RLE(A) = a_1^{M_1} \cdots a_m^{M_m}$ and $\RLE(B) = b_1^{N_1} \cdots b_n^{N_n}$.
  The {\em compressed DP table (cDP table) of $A, B$} is an $O(mN+nM)$-space compressed representation of the DP table of $A,B$
  which holds only the values of the DP table for $(M_{1..p},j)$ and $(i,N_{1..q})$,
  where, $1 \leq i \leq M$, $1 \leq j \leq N$, $1 \leq p \leq m$, $1 \leq q \leq n$, $M_{1..p} = M_1 + \cdots + M_p$, $N_{1..q} = N_1 + \cdots + N_q$.
\end{definition}

\begin{figure}[htbp]
  \begin{center}
    \begin{tabular}{|c|c|cccc|ccc|cc|} \hline
      \multicolumn{2}{|c|}{} & \multicolumn{9}{c|}{$B$} \\ \cline{3-11}
      \multicolumn{2}{|c|}{} & \tt{a} & \tt{a} & \tt{a} & \tt{a} & \tt{b} & \tt{b} & \tt{b} & \tt{a} & \tt{a} \\ \hline
        & \tt{b} &   &   &   & 0 &   &   & 1 &   & 1 \\
        & \tt{b} &   &   &   & 0 &   &   & 2 &   & 2 \\
        & \tt{b} & 0 & 0 & 0 & 0 & 1 & 2 & 3 & 3 & 3 \\ \cline{2-11}
    $A$ & \tt{a} &   &   &   & 1 &   &   & 3 &   & 4 \\
        & \tt{a} &   &   &   & 2 &   &   & 3 &   & 5 \\
        & \tt{a} &   &   &   & 3 &   &   & 3 &   & 5 \\
        & \tt{a} & 1 & 2 & 3 & 4 & 4 & 4 & 4 & 4 & 5 \\ \hline
    \end{tabular}
  \end{center}
  \caption{An example of a compressed $\Lpref$ DP table for strings $A=\mathtt{bbbaaaa}$ and $B=\mathtt{aaaabbbaa}$.}
  \label{tbl:compressedDPtable}
\end{figure}

Figure~\ref{tbl:compressedDPtable} illustrates the values stored in the cDP table for strings
$A=\mathtt{bbbaaaa}$, $B=\mathtt{aaaabbbaa}$.
Note that although the figure depicts a sparsely filled table of size $M\times N$,
the values are actually stored in two (completely filled) tables:
one of size $m\times N$, holding the values of $(M_{1..p},j)$, and another of size $M\times n$,
holding the values of $(i, N_{1..q})$, for a total of $O(mN+nM)$ space.
Below are results adapted from~\cite{DBLP:journals/ipl/BunkeC95} we will use.

\begin{lemma}[{\cite[Theorem 7]{DBLP:journals/ipl/BunkeC95}}]
  \label{lem:CDPtable}
  Let $A$ and $B$ be any strings where $|A| = M$, $|B| = N$, $|\RLE(A)| = m$ and $|\RLE(B)| = n$.
  The compressed DP table of $A$ and $B$ can be computed in $O(mN + nM)$ time and space.
\end{lemma}

\begin{lemma}[{\cite[Lemma 3]{DBLP:journals/ipl/BunkeC95}}]
  \label{lem:LCScalculation1}
  Let $\alpha \in \Sigma$ and
  let $A$ and $B$ be any strings where $|A| = M$ and $|B| = N$.
  For any integer $d \geq 1$,
  if $A[M-d+1..M] = B[N-d+1..N] = \alpha^d$,
  then $\Lpref(M,N) = \Lpref(M-d,N-d) + d$.
\end{lemma}

\begin{lemma}[{\cite[Lemma 5]{DBLP:journals/ipl/BunkeC95}}]
  \label{lem:LCScalculation2}
  Let $\alpha, \beta \in \Sigma$, $\alpha \neq \beta$ and
  let $A$ and $B$ be any strings where $|A| = M$ and $|B| = N$.
  For any integers $d \geq 1$ and $d' \geq 1$,
  if $A[M-d+1..M] = \alpha^d$ and $B[N-d'+1..N] = \beta^{d'}$
  then $\Lpref(M,N) = \max\{\Lpref(M-d,N), \Lpref(M,N-d')\}$.
\end{lemma}

From Lemmas~\ref{lem:LCScalculation1} and~\ref{lem:LCScalculation2},
we easily obtain the following Lemma~\ref{lem:DPandCDPtable}.

\begin{lemma}\label{lem:DPandCDPtable}
  Let $A$ and $B$ be any strings.
  Any entry of the DP table of $A$ and $B$ can be retrieved in $O(1)$ time
  by using the compressed DP table of $A$ and $B$.
\end{lemma}

From Lemma~\ref{lem:CDPtable}, we can compute in $O(mN+nM)$ time,
two cDP tables of $A,B$ which respectively hold the values of
$\Lpref(M_{1..p},j)$, $\Lpref(i,N_{1..q})$
and
$\Lsuf(M_{1..p},j)$, $\Lsuf(i,N_{1..q})$,
each of them taking $O(mN+nM)$ space.
From Lemma~\ref{lem:DPandCDPtable},
we can obtain $\Lpref(i,j)$, $\Lsuf(i,j)$ for any $i$ and $j$ in $O(1)$ time.
Actually, to make Lemma~\ref{lem:DPandCDPtable} work, we also need to be able to convert the indexes
between DP and cDP in constant time, i.e.,
for any $1 \leq p \leq m$, $1 \leq q \leq n$, the values $M_{1..p}$ and $N_{1..q}$, and
for any $1 \leq i \leq M$, $1 \leq j \leq N$, the largest $p,q$ such that $M_{1..p} \leq i$, $N_{1..q} \leq j$.
This is easy to do by preparing some arrays in $O(M+N)$ time and space.

Now we are ready to show the running time of our algorithm for the case $|\RLE(C)|>1$.
We can compute $\RLE(A)$,$\RLE(B)$,$\RLE(C)$ from $A,B,C$ in $O(M+N+K)$ time.
In Step 1, we have from Lemma~\ref{lem:MinInterval k>1}, that the number of all
minimal $C$-intervals of $A,B$ are respectively $O(m)$ and $O(n)$, and can be computed in
$O(M + N + mK + nK)$ time.
For the preprocessing of Step 2, we build the cDP tables holding the values of
$\Lpref(i,j)$, $\Lsuf(i,j)$ for $1 \leq i \leq M$, $1 \leq j \leq N$,
which can be computed in $O(mN+nM)$ time and space from Lemma~\ref{lem:CDPtable}.
With these tables, we can obtain for any
$i,j$, the values $\Lpref(i,j)$, $\Lsuf(i,j)$ in constant time from Lemma~\ref{lem:DPandCDPtable}.
Since there are $O(mn)$ pairs of a minimal $C$-interval of $A$ and a minimal $C$-interval of $B$,
the total time for Step 2, i.e. computing $\Lpref$ and $\Lsuf$ for each of the pairs, is $O(mn)$.
Since $n\leq N, m\leq M$, and we can assume that $K \leq M, N$, the total time is $O(mN+nM)$. 
Thus Theorem~\ref{thm:TimeComplexity k>1} holds.

\subsubsection{Case $|\RLE(C)| = 1$}

Next, we consider the case where $|\RLE(C)| = 1$, and $C$ consists of only one run.

\begin{theorem}\label{thm:TimeComplexity k=1}
  Let $A,B,C$ be any strings and let $|A|=M$, $|B|=N$, $|\RLE(A)|=m$ and $|\RLE(B)|=n$.
  If $|\RLE(C)|=1$, we can compute the length of an STR-IC-LCS of $A,B,C$ in $O(mN + nM)$ time.
\end{theorem}

For Step 1, we compute all minimal $C$-intervals of $A$ and $B$ by
Lemma~\ref{lem:MinInterval k=1}. Note the difference from Lemma~\ref{lem:MinInterval k>1}
in the case of $|\RLE(C)|>1$.

\begin{lemma}\label{lem:MinInterval k=1}
  If $|\RLE(C)|=1$,
  the number of minimal $C$-intervals of $A$ and $B$
  are $O(M)$ and $O(N)$, respectively,
  and these can be enumerated in $O(M)$ and $O(N)$ time, respectively.
\end{lemma}

\begin{proof}
  Let $\alpha \in \Sigma$, $C=\alpha^K$, and let $M_\alpha$ be the number of times that $\alpha$ occurs in $A$.
  Then the number of minimal $C$-intervals of $A$ is $M_\alpha-K+1 \in O(M)$.
  The minimal $C$-intervals can be enumerated in $O(M)$ time by checking all positions of $\alpha$ in $A$.
  The same applies to $B$.
\end{proof}

From Lemma~\ref{lem:MinInterval k=1}, we can see that the number of pairs
of minimal $C$-intervals of $A$ and $B$ can be $\Theta(MN)$, and we cannot
afford to consider all of those pairs for Step 2.
We overcome this problem as follows.
Let
$U = \{[s_1, f_1], \ldots, [s_l, f_l]\}$ be the
set of all minimal $C$-intervals of $A$.
Consider the partition $G(1), \ldots, G(g)$ of $U$
which are the equivalence classes induced by the following equivalence relation on $U$:
For any $1 \leq p \leq q \leq m$ and $[s_x, f_x], [s_y, f_y] \in U$,
\begin{eqnarray}
[s_x,f_x] \equiv [s_y,f_y] \iff
M_{1..p-1} < s_x ,s_y \leq M_{1..p} \mbox{ and }
M_{1..q-1} < f_x , f_y \leq M_{1..q},\label{eqn:equivalence}
\end{eqnarray}
where, $M_{1..0} = 0$. In other words, $[s_x,f_x]$ and $[s_y,f_y]$ are in the same equivalence class
if they start in the same run, and end in the same run.
Noticing that minimal $C$-intervals cannot be completely contained in another,
we can assume that for $1\leq h < h'\leq g$, $[s_x, f_x] \in G(h)$ and $[s_y, f_y] \in G(h')$,
we have $s_x < s_y$ and $f_x < f_y$.

\begin{lemma}\label{lem:numofG}
  Let $G(1), \ldots, G(g)$ be the partition of the set $U$ of all minimal $C$-intervals of $A$
  induced by the equivalence relation~(\ref{eqn:equivalence}).
  Then, $g \in O(m)$.
\end{lemma}

\begin{proof}
  Let $1 \leq x < y \leq l$ and $2 \leq h \leq g$.
  For any $[s_x, f_x] \in G(h-1)$ and $[s_y, f_y] \in G(h)$,
  let $1 \leq p \leq q \leq m$ satisfy
  $M_{1..p-1} < s_x \leq M_{1..p}$, $M_{1..q-1} < f_x \leq M_{1..q}$.
  Since the intervals are not equivalent, either $M_{1..p} < s_y$ or $M_{1..q} < f_y$ must hold.
  Thus, $g \in O(m)$.
\end{proof}

Equivalently for $B$, we consider the set $U' = \{[s'_1, f'_1], \ldots, [s'_{l'}, f'_{l'}]\}$ 
of all minimal $C$-intervals of $B$, and the partition $G'(1), \ldots, G'(g')$ of $U'$ 
based on the analogous equivalence relation, where $g' \in O(n)$.

For some $h$, let $[s_x, f_x]$, $[s_y, f_y]$ be the minimal $C$-intervals in $G(h)$ with the smallest and largest start positions.
Since by definition, $A[s_x] = \cdots = A[s_y] = A[f_x] = \cdots = A[f_y]$,
we have $G(h) = \{[s_x+i, f_x+i], [s_x+1, f_x+1], \ldots, [s_y, f_y]\}$.
The same can be said for $G'(h')$ of $B$.
From this observation, we can show the following Lemma~\ref{lem:calcLCSinG}.

\begin{lemma}\label{lem:calcLCSinG}
  For any $1 \leq h \leq g$ and $1 \leq h' \leq g'$, let
  $[s, f]$, $[s+d, f+d] \in G(h)$ and 
  $[s', f']$, $[s'+d, f'+d] \in G'(h')$,
   for some positive integer $d$.
  Then,
  \[\Lpref(s-1, s'-1) + \Lsuf(f+1, f'+1) = \Lpref(s+d-1, s'+d-1) + \Lsuf(f+d+1, f'+d+1).\]
\end{lemma}

\begin{proof}
  Since $A[s..s+d] = A[f..f+d] = B[s'..s'+d] = B[f'..f'+d] = C[1]^d$,
  we have from Lemma~\ref{lem:LCScalculation1},
  $\Lpref(s+d-1, s'+d-1) = \Lpref(s-1, s'-1) + d$,  
  and
  $\Lsuf(f+1, f'+1) = \Lsuf(f+d+1, f'+d+1) + d$.
\end{proof}

From Lemma~\ref{lem:calcLCSinG}, we can see that for any
$G(h), G'(h')$ $(1 \leq h \leq g$, $1 \leq h' \leq g')$,
we do not need to compute $\Lpref(s-1, s'-1) + \Lsuf(f+1, f'+1)$ for all pairs of
$[s, f] \in G(h)$ and $[s', f'] \in G'(h')$.
Let $\Gmin(h)$ and $\Gmin'(h')$ be the minimal $C$-intervals respectively in $G(h)$ and $G'(h')$ with the smallest starting position.
Then, we only need to consider the combination of $\Gmin(h)$ with each of $[s', f'] \in G'(h')$,
and the combination of each of $[s, f] \in G(h)$ with $\Gmin'(h')$.
Therefore, of all combinations of minimal $C$-intervals in $U$ and $U'$,
we only need to consider for all $1 \leq h \leq g$ and $1 \leq h' \leq g'$,
the combination of $\Gmin(h)$ with each of $U'$, and
each of $U$ with $\Gmin'(h')$.
The number of such combinations is clearly $O(mN + nM)$.

For example, consider
$\RLE(A) = \mathtt{a}^5\mathtt{b}^3\mathtt{a}^4\mathtt{b}^2\mathtt{a}^1$,
$\RLE(B) = \mathtt{a}^1\mathtt{b}^3\mathtt{a}^7\mathtt{b}^3$,
$\RLE(C) = \mathtt{a}^5$.
For the minimal $C$-intervals of $A$, we have
$G(1) = \{[1,5]\}$,
$G(2) = \{[2,9],[3,10],[4,11],[5,12]\}$,
$G(3) = \{[9,15]\}$.
For the minimal $C$-intervals of $B$, we have
$G'(1) = \{[1,8]\}$,
$G'(2) = \{[5,9],[6,10],[7,11]\}$.
Also, $\Gmin(2) = [2,9]$, $\Gmin'(2) = [5,9]$.
Figure~\ref{tbl:LCS k=1} shows the lengths of the LCS of prefixes and suffixes for each combination between
minimal $C$-intervals in $G(2)$ and $G'(2)$.
The gray part is the values that are referred to.
The values denoted inside parentheses are not stored in the cDP table, but each of them can be
computed in $O(1)$ time from Lemma~\ref{lem:DPandCDPtable}.
Figure~\ref{tbl:LCSsum k=1} shows the sum of the LCS of prefixes and suffixes corresponding to the gray part.
Due to Lemma~\ref{lem:calcLCSinG}, the values along the diagonal are equal.
Thus, for the combinations of minimal $C$-intervals in $G(2)$, $G'(2)$, we only need to consider
the six combinations: $([2,9],[5,9])$,$([2,9],[6,10])$,$([2,9],[7,11])$,$([3,10],[5,9])$,$([4,11],[5,9])$,$([5,12],[5,9])$.

\begin{figure}[htbp]
  \begin{center}
    \caption{An example depicting the LCSs of corresponding prefixes (left) and suffixes (right) of all combinations of $G(2)$ and $G'(2)$ for strings $\RLE(A) = \mathtt{a}^5\mathtt{b}^3\mathtt{a}^4\mathtt{b}^2\mathtt{a}^1$,
$\RLE(B) = \mathtt{a}^1\mathtt{b}^3\mathtt{a}^7\mathtt{b}^3$, and
$\RLE(C) = \mathtt{a}^5$. The values denoted inside parentheses are not stored in the cDP table, but each of them can be
computed in $O(1)$ time.}
    \label{tbl:LCS k=1}
    \begin{tabular}{cc}
      \begin{minipage}{0.5\hsize}
        \begin{center}
          \begin{tabular}{|c|c|c|ccc|cccc|} \hline
            \multicolumn{2}{|c|}{} & \multicolumn{8}{c|}{$B$} \\ \cline{3-10}
            \multicolumn{2}{|c|}{} & \tt{a} & \tt{b} & \tt{b} & \tt{b} & \tt{a} & \tt{a} & \tt{a} & $\cdots$ \\ \hline
                & \tt{a} & 1 &   &   & \tgray1 & \tgray(1) & \tgray(1) &   & \\
                & \tt{a} & 1 &   &   & \tgray1 & \tgray(2) & \tgray(2) &   & \\
                & \tt{a} & 1 &   &   & \tgray1 & \tgray(2) & \tgray(3) &   & \\
            $A$ & \tt{a} & 1 &   &   & \tgray1 & \tgray(2) & \tgray(3) &   & \\
                & \tt{a} & 1 & 1 & 1 & 1 & 2 & 3 & 4 & $\cdots$ \\ \cline{2-10}
                & \tt{b} & 1 &   &   & 2 &   &   &   & \\
                & $\vdots$ & $\vdots$ &   &   & $\vdots$ &   &   &   & \\ \hline
          \end{tabular}
        \end{center}
      \end{minipage}
      \begin{minipage}{0.5\hsize}
        \begin{center}
          \begin{tabular}{|c|c|cccc|ccc|} \hline
            \multicolumn{2}{|c|}{} & \multicolumn{7}{c|}{$B$} \\ \cline{3-9}
            \multicolumn{2}{|c|}{} & $\cdots$ & \tt{a} & \tt{a} & \tt{a} & \tt{b} & \tt{b} & \tt{b} \\ \hline
                & $\vdots$ &   &   &   &   & $\vdots$ &   & \\
                & \tt{a} &   &   & \tgray(4) & \tgray(3) & \tgray2 &   & \\
                & \tt{a} &   &   & \tgray(4) & \tgray(3) & \tgray2 &   & \\
            $A$ & \tt{a} &   &   & \tgray(3) & \tgray(3) & \tgray2 &   & \\ \cline{2-9}
                & \tt{b} & $\cdots$ & 2 & \tgray2 & \tgray2 & \tgray2 & 2 & 1 \\ 
                & \tt{b} &   &   &   &   & 1 &   & \\ \cline{2-9}
                & \tt{a} & $\cdots$ & 1 & 1 & 1 & 0 & 0 & 0 \\ \hline
          \end{tabular}
        \end{center}
      \end{minipage}
    \end{tabular}
  \end{center}
\end{figure}

\begin{figure}[htbp]
  \begin{center}
    \caption{Sum of the lengths of LCSs of corresponding prefixes and suffixes shown in Figure~\ref{tbl:LCS k=1}.
      Values along the diagonal are equal (each value is equal to the value to its upper left/lower right).}
    \label{tbl:LCSsum k=1}
    \begin{tabular}{ccc}
      \rowcolor[gray]{0.85} 5 & 4 & 3 \\
      \rowcolor[gray]{0.85} 5 & 5 & 4 \\
      \rowcolor[gray]{0.85} 4 & 5 & 5 \\
      \rowcolor[gray]{0.85} 3 & 4 & 5 \\
    \end{tabular}  
  \end{center}
\end{figure}

Now, we are ready to show the running time of our algorithm for the case $|\RLE(C)|=1$.
We can compute $\RLE(A)$, $\RLE(B)$, $\RLE(C)$ from $A,B,C$ in $O(M+N+K)$ time.
There are respectively $O(M)$ and $O(N)$ minimal $C$-intervals of $A$ and $B$, and 
each of them can be assigned to one of the $O(m)$ and $O(n)$ equivalence classes
$G, G'$, in total of $O(M+N)$ time.
The preprocessing for the cDP table is the same as for the case of $|\RLE(C)|>1$,
which can be done in $O(mN + nM)$ time.
By Lemma~\ref{lem:calcLCSinG}, we can reduce the number of combinations of 
minimal $C$-intervals to consider to $O(mN + nM)$.
Finally, from Lemma~\ref{lem:DPandCDPtable}, the LCS lengths for each combination can be
computed in $O(1)$ using the cDP table.
Therefore, the total running time is $O(mN + nM)$,
proving Theorem~\ref{thm:TimeComplexity k=1}.

From Theorems~\ref{thm:TimeComplexity k>1} and~\ref{thm:TimeComplexity k=1}, the following Theorem~\ref{thm:TimeComplexity} holds.
The pseudo-code for our proposed algorithm is shown in Algorithm~\ref{alg:STR-IC-LCS} (written in Appendix).

\begin{theorem}\label{thm:TimeComplexity}
  Let $A,B,C$ be any strings and let $|A|=M$, $|B|=N$, $|\RLE(A)|=m$ and $|\RLE(B)|=n$.
  We can compute the length of an STR-IC-LCS of $A,B,C$ in $O(mN + nM)$ time.
\end{theorem}

Although we only showed how to compute the length of an STR-IC-LCS, we note that
the algorithm can be modified so as to obtain a RLE of an STR-IC-LCS in $O(m+n)$ time,
provided that $\RLE(C)$ is precomputed, simply by
storing the minimal $C$-intervals $[s,f]$, $[s',f']$, respectively of $A$ and $B$,
that maximizes $\Lpref(s-1,s'-1) + \Lsuf(f+1,f'+1) + |C|$.
From Lemmas~\ref{lem:LCScalculation1} and~\ref{lem:LCScalculation2},
we can
simulate a standard back-tracking of the DP table for obtaining LCSs
with the cDP table to obtain RLE of the LCSs in $O(m+n)$ time.
Finally an RLE of STR-IC-LCS can be obtained by combining the three RLE strings (the two LCSs with $\RLE(C)$ in the middle),
appropriately merging the boundary runs if necessary.

%% file: 4_conclusion.tex
\section{Conclusion}
In this work, we proposed a new algorithm to solve the STR-IC-LCS problem using RLE representation.
We can compute the length of an STR-IC-LCS of strings $A,B,C$ in $O(mN+nM)$ time and space using this algorithm, where $|A| = M$, $|B| = N$, $|RLE(A)| = m$ and $|RLE(B)| = n$.
This result is better than Deorowicz's $O(MN)$ time and space~\cite{DBLP:journals/ipl/Deorowicz12}, which doesn't use RLE.
If we want to know not only the length but also an STR-IC-LCS of $A,B,C$, we can retrieve it in $O(m+n)$ time.

%% file: 5_appendix.tex
\section{Appendix}
Here, we show the pseudo-code for our proposed algorithm.

\begin{algorithm2e}
  \caption{Proposed $O(mN+Mn)$ time algorithm for STR-IC-LCS}
  \label{alg:STR-IC-LCS}
  \KwIn{strings $A$, $B$ and $C$}
  \KwOut{length of an STR-IC-LCS of $A,B,C$}

  \tcp{$[s_x,f_x]$ : a minimal $C$-interval in $A$}
  \tcp{$[s'_y,f'_y]$ : a minimal $C$-interval in $B$}
  \tcp{$l,l'$ : number of minimal $C$-intervals in $A,B$ respectively}
  \tcp{$\Gmin(h),\Gmin'(h')$ : minimum element in $G(h),G'(h')$ respectively}
  \tcp{$g,g'$ : number of sets $G,G'$ respectively}
  
  Make compressed DP tables of $A$ and $B$.\;

  \uIf{$|\RLE(C)| > 1$}{
    Compute all minimal $C$-intervals $[s_1,f_1], \ldots, [s_l,f_l]$ of $A$
    and $[s'_1,f'_1], \ldots, [s'_{l'},f'_{l'}]$ of $B$.
    (use Algorithm~\ref{alg:minCinterval})\;
    
    $\Lmax \gets 0$\;
    \For{$x=1$ {\bf to} $l$}{
      \For{$y=1$ {\bf to} $l'$}{
        $\Lsum \gets \Lpref(s_x-1,s'_y-1) + \Lsuf(f_x+1,f'_y+1)$\;
        \lIf{$\Lmax < \Lsum$}{
            $\Lmax \gets \Lsum$\;
        }
      }
    }
  }
  \Else{
    $l \gets 1-K$;
    $g \gets 1$;
    $\Gmin(1) \gets 1$\;
    \For{$p=1$ {\bf to} $m$}{
      \If{$a_p = C[1]$}{
        \For{$p'=1$ {\bf to} $M_p$}{
          $l \gets l+1$;
          $s_{l+K} \gets M_{1..p}+p'$\;
          \lIf{$l \geq 1$}{
            $f_l \gets M_{1..p}+p'$\;
          }
          \If{$l \geq 2$}{
            \lIf{$s_{l-1}+1 \neq s_l$ or $f_{l-1}+1 \neq f_l$}{
              $g \gets g+1$;
              $\Gmin(g) \gets l$\;
            }
          }
        }
      }
    }
    
    $l' \gets 1-K$;
    $g' \gets 1$;
    $\Gmin'(1) \gets 1$\;
    \For{$q=1$ {\bf to} $n$}{
      \If{$b_q = C[1]$}{
        \For{$q'=1$ {\bf to} $N_q$}{
          $l' \gets l'+1$;
          $s'_{l'+K} \gets N_{1..q}+q'$\;
          \lIf{$l' \geq 1$}{
            $f'_{l'} \gets N_{1..q}+q'$\;
          }
          \If{$l' \geq 2$}{
            \lIf{$s'_{l'-1}+1 \neq s'_{l'}$ or $f'_{l'-1}+1 \neq f'_{l'}$}{
              $g' \gets g'+1$;
              $\Gmin'(g') \gets l'$\;
            }
          }
        }
      }
    }
    $\Gmin(g+1) \gets l+1$;
    $\Gmin'(g'+1) \gets l'+1$\;
    
    $\Lmax \gets 0$\;
    \For{$h=1$ {\bf to} $g$}{
      \For{$h'=1$ {\bf to} $g'$}{
        \For{$x=\Gmin(h)$ {\bf to} $\Gmin(h+1)-1$}{
          $\Lsum \gets \Lpref(s_x-1,s'_{\Gmin'(h')}-1) + \Lsuf(f_x+1,f'_{\Gmin'(h')}+1)$\;
          \lIf{$\Lmax < \Lsum$}{
            $\Lmax \gets \Lsum$\;
          }
        }
        \For{$y=\Gmin'(h')$ {\bf to} $\Gmin'(h'+1)-1$}{
          $\Lsum \gets \Lpref(s_{\Gmin(h)}-1,s'_y-1) + \Lsuf(f_{\Gmin(h)}+1,f'_y+1)$\;
          \lIf{$\Lmax < \Lsum$}{
            $\Lmax \gets \Lsum$\;
          }
        }
      }
    }
  }
  \Return{$\Lmax + K$}\;
\end{algorithm2e}